\documentclass[a4paper,12pt]{article}
\usepackage{authblk}
\usepackage[a4paper,includehead,includefoot,bindingoffset=2cm,margin=2cm]{geometry}

\usepackage[dvipsnames,usenames]{color}

\usepackage{amsmath,amssymb, amsthm}
\usepackage{mathrsfs}
\usepackage{psfrag}
\usepackage{graphicx}
\usepackage{scalefnt}
\usepackage{enumitem}
\usepackage[round]{natbib}
\usepackage{helvet}
\usepackage{ulem}
\usepackage{pdflscape}
\usepackage{algorithmic}
\usepackage[section]{algorithm}
\usepackage[doublespacing]{setspace}

\bibpunct{(}{)}{;}{a}{}{,}

\numberwithin{equation}{section}

\DeclareMathOperator\sP{P}   %probability, treated as an operator
\DeclareMathOperator\diag{diag}
\newcommand{\cvgas}{\xrightarrow{\text{\upshape\tiny a.s.}}}   %convergence almost surely
\newcommand{\integers}{\mathbb{Z}}

\newcommand{\argmax}{\mathop{\rm arg\,max}}

\makeatletter
\def\imod#1{\allowbreak\mkern10mu({\operator@font mod}\,\,#1)}
\makeatother

\newtheorem{thm}{Theorem}[section]
\newtheorem{prop}{Proposition}[section]
\newtheorem{lem}{Lemma}[section]
\newtheorem{cor}{Corollary}[section]

\theoremstyle{definition}
\newtheorem{remark}{Remark}[section]

\newtheorem{dfn}{Definition}[section]
\newtheorem{example}{Example}[section]

\makeatletter
\def\appendix{\renewcommand{\thesection}{\thechapter.\Alph{section}}
	\@ifstar\unnumberedappendix\numberedappendix}
\def\numberedappendix{\@ifnextchar[%]
  \numberedappendixwithtwoarguments\numberedappendixwithoneargument}
\def\unnumberedappendix{\@ifnextchar[%]
  \unnumberedappendixwithtwoarguments\unnumberedappendixwithoneargument}
\def\numberedappendixwithoneargument#1{\numberedappendixwithtwoarguments[#1]{#1}}
\def\unnumberedappendixwithoneargument#1{\unnumberedappendixwithtwoarguments[#1]{#1}}
\def\numberedappendixwithtwoarguments[#1]#2{%
  \ifhmode\par\fi
  \removelastskip
  \vskip 2ex\goodbreak
  \refstepcounter{section}%
  \begingroup
  \noindent\leavevmode\Large\bfseries\raggedright 
  Appendix\ \thesection\quad#2\par
  \endgroup
  \vskip 1ex\nobreak
  \addcontentsline{toc}{section}{%
    \protect\numberline{\thesection}%
    #1}%
  }
\def\unnumberedappendixwithtwoarguments[#1]#2{%
  \ifhmode\par\fi
  \removelastskip
  \vskip 2ex\goodbreak
  \begingroup
  \noindent\leavevmode\Large\bfseries\raggedright 
  #2\par
  \endgroup
  \vskip 1ex\nobreak
  \addcontentsline{toc}{section}{%
    #1}%
  }
\makeatother

\begin{document}
\title{Full Reconstruction of Non-Stationary Strand-Symmetric Models on Rooted Phylogenies}
\author{Benjamin D Kaehler}
\affil{Research School of Biology, Australian National University, Canberra, Australian Capital Territory, Australia}
\maketitle
\begin{abstract}
Understanding the evolutionary relationship among species is of fundamental importance to the biological sciences. The location of the root in any phylogenetic tree is critical as it gives an order to evolutionary events. None of the popular models of nucleotide evolution used in likelihood or Bayesian methods are able to infer the location of the root without exogenous information. It is known that the most general Markov models of nucleotide substitution can also not identify the location of the root or be fitted to multiple sequence alignments with less than three sequences. We prove that the location of the root and the full model can be identified and statistically consistently estimated for a non-stationary, strand-symmetric substitution model given a multiple sequence alignment with two or more sequences. We also generalise earlier work to provide a practical means of overcoming the computationally intractable problem of labelling hidden states in a phylogenetic model.
\end{abstract}


\section{Introduction}

The location of the root in a molecular phylogeny has contributed to criminal convictions \citep{gonzalez2013molecular}, been used to understand the source and epidemiology of human viruses \citep{podsiadlo2013molecular}, determined how biodiversity conservation resources were distributed \citep{faith2006phylogenetic}, been used to develop potential HIV vaccines \citep{nickle2003consensus}, and played an important role in our understanding of the tree of life \citep{murphy2007using}. While an unrooted phylogenetic tree can be used to infer relatedness between species, without the location of the root we know nothing of the order of evolutionary events. It might then be surprising that none of the commonly used Markov models of nucleotide or codon substitution can be used to identify the location of the root without incorporating information that is exogenous to the model. The class of Markov models to which we refer will be made precise in the next section.

ModelTest is one of the most popular pieces of software for selecting phylogenetic models of character substitution \citep{posada2008jmodeltest}. It allows users to determine which of 88 time-reversible (hereafter {\it reversible}) substitution models best fits their data. By definition, for a reversible model the location of the root in a phylogeny cannot change the probability distribution of the observed data, the column frequencies. Some software \citep{knight:2007:00} allows users to fit non-stationary models of character substitution such as that of \citet{barry1987statistical}. Unfortunately, the theoretical results that exist around fully general models \citep{chang:1996:00} explicitly state that for such models the location of the root is not statistically identifiable, that one cannot use such models to ask where on an edge a root resides, and that it is possible to reformulate the model so that any node is the root.

In practice, the location of the root is usually determined by declaring that a specific taxon in a phylogeny is an {\it outgroup} or by making a {\it molecular clock} assumption \citep{felsenstein2004inferring}. The first method assumes that the location of the root is already known, that it is on the edge connected to the outgroup. The second method comes in varying degrees of complexity. In its most simple form it assumes that the tree is {\it ultrametric}, that the genetic distance from the root node to each tip is identical. In more sophisticated Bayesian approaches the location of the root enters the calculation as part of the prior distribution of tree topologies and branch lengths, so that the tree is not necessarily ultrametric but that in some sense the evolutionary time from the root of the tree to the tips is the same along every lineage \citep{drummond2007beast}.

A third method is to use a substitution model that is able to identify the location of the root. Such a model must not be reversible, but using a model that is not reversible does not automatically ensure that the model is able to identify the root. This statement is easily justified using the findings in \citet{chang:1996:00}, where a model that is non-stationary, so also non-reversible, is not able to recover the root. That a non-reversible model might not be even theoretically able to discover the root seems to have been missed by some authors.

\citet{yang1995use} fitted a non-stationary model to rooted topologies of real data using maximum likelihood and found that the location of the root of the tree had a significant effect on likelihood estimates. This is useful empirical evidence but the authors made no attempt to prove that their model is identifiable.

\citet{huelsenbeck2002inferring} fitted a non-reversible but stationary model to real and simulated data and found that while the outgroup and molecular clock methods were able to recover the location of the root in many cases, their model was not. Again, they made no effort to show that their model is theoretically capable of recovering the location of the root, so the poor performance of their model is not necessarily a reflection on the ability of all substitution models to recover the location of the root.

\citet{yap2005rooting} systematically reproduced the results in \citet{yang1995use} and \citet{huelsenbeck2002inferring} and, with some small discrepancies with the earlier studies, found again that a non-stationary model was able to make statistically significant inferences about the location of the root, but that a non-reversible model did little better than a reversible model. This also was empirical research that left unanswered questions about whether any of the models were theoretically able to identify the location of the root.

The contribution of the present work is to constructively prove that there is a non-stationary substitution process that identifies the location of the root that can be statistically consistently estimated from data. Indeed, the model is shown to be consistently estimable for two taxa. This is not possible for general non-stationary processes \citep{chang:1996:00,bonhomme2014nonparametric}, so we make the additional assumption that the process is {\it strand-symmetric}; that the process of evolution is identical on the sense and antisense strands of DNA. The conditions of the proof ensure that the process is non-stationary, and we show that a non-reversible, stationary model is not identifiable so cannot be consistently estimated. This observation sheds some light on the success of non-stationary processes and the failure of non-reversible, stationary processes at detecting the root in the literature.

Much has been written about the biological mechanisms that result in nucleotide substitution processes being strand asymmetric and there is now substantial empirical evidence to support strand asymmetry's existence in nature. \citet{touchon2008gc} provide a good review of the subject. Strand asymmetry seems to be a localised phenomenon, existing on the scale of genes rather than genomes, and appears to be common in prokaryote and organelle genomes but not in eukaryotes. Nucleotide compositional asymmetry is the most common measure used for statistical inference. Under very loose assumptions \citep{lobry1995properties,lobry1999evolution} a strand-symmetric process should result in the proportions of As and Ts being equal and the proportions of Gs and Cs being equal on a single strand. Strand asymmetry can also be inferred by directly comparing estimated rates of nucleotide substitution, although most of the evidence seems to come from counting substitutions in ancestral state reconstructions based on maximum parsimony \citep[eg.][]{wu1987inequality,bulmer1991strand,francino2000strand}.

Strand-symmetric models have been used in a maximum likelihood context, although rarely for the purpose of establishing whether strand symmetry is a reasonable assumption. \citet{yap2004identification} comment that the reversible models that they fitted to real data seemed to exhibit strand symmetry. \citet{squartini2008quantifying} fitted a continuous-time, non-stationary, strand-symmetric model on a known, rooted, four-taxon topology to two genome-scale data sets. They comment that their model is not identifiable on the edges incident to the root, but that it is identifiable on the other two edges. This is not strictly correct. As stated in \citet[Remark~2]{chang:1996:00}, the labelling of states at the internal nodes is not automatically identifiable even for a continuous-time model. Also, as the mapping from the discrete-time process considered in \citet{chang:1996:00} to the continuous-time process fitted in \citet{squartini2008quantifying} is not always unique \citep[Section~2.3]{higham:2008:00}, continuous-time models are not identifiable under the results in \citet{chang:1996:00} without further constraints.

Strand-symmetric substitution models have also been approached from a theoretical perspective \citep{casanellas2005strand, jarvis2013matrix} in the context of reversible models. \citet{jarvis2013matrix} make the observation that strand-symmetric models enjoy the property of {\it closure}; that a model which is strand-symmetric on two adjacent phylogenetic branches is strand-symmetric across the two branches as well. It is straightforward to show that this property extends to the non-stationary processes that we consider here.

The proofs in this work mirror those in \citet{chang:1996:00}, but apart from adding the assumption of strand symmetry and removing the assumption of an unrooted tree, we also relax an important assumption that \citet{chang:1996:00} makes about the structure of the model. As noted in \citet{zou2011the-parameters} and addressed in \citet{mossel2006learning} the model of \citet{barry1987statistical} is only identifiable up to an arbitrary relabelling of states at internal nodes of the phylogeny. \citet{chang:1996:00} addresses this problem by assuming that the transition probability matrix for every edge in the topology is {\it reconstructible from rows}. As we shall demonstrate, this assumption can be restrictive in practice, so, motivated by Remark~4 in \citet{chang:1996:00}, we partially relax it.

In Section~\ref{notation} we briefly introduce the necessary notation and theoretical context. Section~\ref{identifiability_of_the_full_model} contains the main results of the paper, where we prove that the full topology and parameters of a discrete-time, non-stationary, strand-symmetric Markov model can be recovered from the pairwise joint probability distributions of states between extant taxa. In this section we also extend the result to continuous-time models which are used more commonly in practice. Section~\ref{consistency_of_the_full_model} proves that the results in Section~\ref{identifiability_of_the_full_model} provide the necessary basis for consistent statistical estimation of the models in question for multiple sequence alignments of increasing length. Section~\ref{conculding_remarks} gives some concluding remarks.

\section{Markov Models on Trees: Definitions and Notation}\label{notation}

We consider a finite set of extant taxa $T$ whose phylogenetic history we wish to infer. The history is modelled as a tree which is a set of nodes $S$ which represent  extant or ancestral species and undirected edges $E$ which represent genetic descent. The edges are a set of unordered pairs of nodes, so if $\{r,s\}\in E$, an edge exists between nodes $r$ and $s$. The nodes consist of the {\it terminal nodes}, which are just $T$, and the {\it internal nodes} $N$, so that $N=S\setminus T$. The internal nodes represent ancestral taxa at branching points.

The {\it degree} of a node is the number of edges incident on that node. We say that a tree is {\it unrooted} if it contains no internal nodes with degree less than three. A tree is {\it rooted} if it contains a single node with degree two, which we call the {\it root}. We do not consider trees with more than one internal node of degree two.

The model of character substitution is properly considered a probabilistic graphical model defined on the tree. That is, we associate a random variable $X_s\in\mathcal{C}$ with each node $s\in S$ so that $\{X_s\}_{s\in S}$ represents the history of a single column in a multiple sequence alignment. Each $X_s$ is independent of all other $X_r$, conditional on the states at the nodes neighbouring $s$. As we will focus on the assumption of strand symmetry, we will assume that $\mathcal{C}=\{\text{A}, \text{C}, \text{G}, \text{T}\}$. We also assume that each column in the alignment is an independent observation of the multivariate random variable $\{X_s\}_{s\in T}$.

The model is then specified by a marginal probability row 4-vector $\pi^s$, where $\pi^s(i)=\sP\left(X_s=i\right)$ for a fixed node $s$, and a $4\times 4$ transition probability matrix $P^{rs}$ for each edge $\{r,s\}\in E$, where $P^{rs}(i,j) = P(X_s=j|X_r=i)$. As we are interested in rooted tree topologies, it is convenient to characterise the model by $\pi^r$ where $r$ is the root and the $P^{st}$ where $\{s,t\}\in E$ and $t$ is further than $s$ from $r$. In this context we will sometimes characterise each $P^{st}$ as the result of a continuous-time Markov substitution process on $\{s,t\}$, in which case $P^{st}=\exp Q^{st}$, where $\exp$ is the matrix exponential and $Q^{st}$ is a transition rate matrix.

The statistical challenge is then, for a fixed set of terminal nodes $T$ and a set of $n$ observations of $\{X_s\}_{s\in T}$, which are the $n$ columns in our multiple sequence alignment, to show that as $n$ tends to infinity our estimates of the tree topology and the probabilistic parameters tend to the true values of the generating model.

\section{Identifiability of the Full Model}\label{identifiability_of_the_full_model}

\subsection{Identifiability with a Rooted Two-Taxon Topology}

It has been known at least since \citet{chang:1996:00} that a full Markov model is not identifiable for a rooted topology with two terminal taxa without additional assumptions. We shall now prove, under mild conditions, that a non-stationary, strand-symmetric model is identifiable for this topology. This result will provide the foundation for showing that rooted topologies of any size are recoverable, and that the full model is identifiable using joint distributions of states between pairs of taxa.

\begin{dfn}\label{sstructure}
We call a matrix $M$ {\it strand-symmetric} if it takes the form 
\begin{align*}
M=\left(\begin{matrix}
\delta &  \alpha & \beta & \gamma \\
\zeta &  \phi & \eta & \theta\\
\theta &  \eta & \phi & \zeta \\
\gamma &  \beta & \alpha & \delta
\end{matrix}\right).
\end{align*}
\end{dfn}

Note that if a transition probability matrix is strand-symmetric in this sense, it only fits the usual definition of strand symmetry if the states are properly ordered. Several such orderings are possible. One is $(A,C,G,T)$.

\begin{dfn}
We say that a probability 4-vector $\pi$ is {\it compositionally asymmetric} if it has non-zero elements and $\pi(1)\neq\pi(4)$,  $\pi(2)\neq\pi(3)$, $\pi(1)/\pi(4)\neq\pi(2)/\pi(3)$, and $\pi(1)/\pi(4)\neq\pi(3)/\pi(2)$.
\end{dfn}

We note that any stationary marginal distribution of a strand-symmetric process violates all four of the conditions of compositional asymmetry, provided that the states are again appropriately ordered as $(A,C,G,T)$ or similar. However, some distributions that are not the stationary distribution of a strand-symmetric process are included in the set of compositionally asymmetric distributions. The reason for this will become apparent in the following proof.

\begin{lem}\label{cherry}Take a two-taxon discrete-time Markov model that is defined by a root distribution $\pi^m=\left(\begin{matrix}\pi^m(1)&\ldots&\pi^m(4)\end{matrix}\right)$ and two $4\times4$ probability transition matrices $P^{ma}$ and $P^{mb}$. Define $\Pi^m=\diag\pi^m$. Assume
\begin{enumerate}[label=(\alph*),ref=\alph*]
\item $P^{ma}$ and $P^{mb}$ are strand-symmetric\label{astsy};
\item $\pi^m$ is compositionally asymmetric\label{aeig}; and
\item $P^{ma}$ and $P^{mb}$ are invertible\label{ainv}.
\end{enumerate}
Then $\pi^m$, $P^{ma}$, and $P^{mb}$ are uniquely determined by the joint probability matrix $J^{ab}=\left(P^{mb}\right)^\intercal\Pi^m P^{ma}$ up to one of eight permutations of the states at node $m$. That is, the model is identifiable up to a suitable reordering of internal states.
\end{lem}

\begin{example}\label{statex}
It is interesting that a non-reversible, stationary, strand-symmetric model is not identifiable for two taxa, as can be shown by counterexample. Take such a model defined by
\begin{align*}
\pi^m &= \left(\begin{matrix}0.20& 0.30 & 0.30 & 0.20\end{matrix}\right)\quad\text{and}\\
P^{ma} &= P^{mb} = \left(\begin{matrix}
0.38 & 0.25 & 0.21 & 0.16 \\
0.16 & 0.43 & 0.27 & 0.14 \\
0.14 & 0.27 & 0.43 & 0.16 \\
0.16 & 0.21 & 0.25 & 0.38 
\end{matrix}\right)
\end{align*}
that yields the joint probability distribution $J^{ab} = \left(P^{mb}\right)^\intercal\Pi P^{ma}$. The model defined by
\begin{align*}
\pi^m &= \left(\begin{matrix}0.20& 0.30 & 0.30 & 0.20\end{matrix}\right),\\
P^{ma} &= \left(\begin{matrix}
0.35 & 0.34 & 0.12 & 0.19 \\
0.15 & 0.46 & 0.24 & 0.15 \\
0.15 & 0.24 & 0.46 & 0.15 \\
0.19 & 0.12 & 0.34 & 0.35
\end{matrix}\right),\quad\text{and}\quad
P^{mb} = \left(\begin{matrix}
0.43 & 0.27 & 0.19 & 0.11 \\
0.07 & 0.39 & 0.31 & 0.23 \\
0.23 & 0.31 & 0.39 & 0.07 \\
0.11 & 0.19 & 0.27 & 0.43 
\end{matrix}\right)
\end{align*}
is also non-reversible, stationary, and strand-symmetric, and yields the same joint probability distribution $J^{ab}$. Both models satisfy all constraints of Lemma~\ref{cherry} except assumption~\ref{aeig}. Generating such examples is straightforward, and a Python script for generating random examples is included in the ancillary material. The above example was generated by that script and rounded to two decimal places.
\end{example}

\begin{proof}[Proof of Lemma~\ref{cherry}]
Single out the permutation matrix
\begin{align*}
S=\left(\begin{matrix}
0 & 0 & 0 & 1\\
0 & 0 & 1 & 0\\
0 & 1 & 0 & 0\\
1 & 0 & 0 & 0
\end{matrix}\right),
\end{align*}
which has the effect of reversing the order of rows or columns, depending on the direction of multiplication.

First note that if a matrix $P$ is strand-symmetric, then $P=SPS$. Also, $P^\intercal = (SPS)^\intercal = SP^\intercal S$ and $P^{-1} = (SPS)^{-1} = SP^{-1}S$.

By assumptions \ref{astsy} and \ref{ainv} and because $SS=I$, 
\begin{align*}
S\left(J^{ab}\right)^{-1}SJ^{ab} &= S \left(P^{ma}\right)^{-1}SS\left(\Pi^m\right)^{-1}\left(\left(P^{mb}\right)^\intercal\right)^{-1} S \left(P^{mb}\right)^\intercal SS \Pi^m P^{ma} \\
&= \left(P^{ma}\right)^{-1} S \left(\Pi^m\right)^{-1} S \Pi^m P^{ma} \\
&= \left(P^{ma}\right)^{-1} \diag\left(\begin{matrix}\frac{\pi^m(1)}{\pi^m(4)}&\frac{\pi^m(2)}{\pi^m(3)}&\frac{\pi^m(3)}{\pi^m(2)}&\frac{\pi^m(4)}{\pi^m(1)}\end{matrix}\right)P^{ma}.
\end{align*}
The final line is an eigendecomposition of $G=S\left(J^{ab}\right)^{-1}SJ^{ab}$ whose left-eigenvectors are the rows of $P^{ma}$. The eigenvectors of $G$ are unique by assumption~\ref{aeig}, so the eigendecomposition is unique up to scaling of the eigenvectors. As the rows of $P^{ma}$ must sum to one, these scaling factors are uniquely determined.

We have shown that the rows of $P^{ma}$ can be uniquely recovered from $J^{ab}$. The order of these rows is not identifiable without further assumptions, so for the moment we will say that $P^{ma}$ can be identified up to a set of one of eight permutations of the rows such that the resulting $P^{ma}$ is strand-symmetric in form. Then, by assumption~\ref{ainv}, $\pi^m=\mathbf{1}J^{ab}\left(P^{ma}\right)^{-1}$ and $\left(P^{mb}\right)^\intercal=J^{ab}\left(\Pi^mP^{ma}\right)^{-1}$, so $\pi^m$ and $P^{mb}$ are also identifiable up to one of eight permutations of their elements and rows that correspond to the permutation chosen for $P^{ma}$.
\end{proof}

\subsection{Identifiability of the Topology}

In this section we build on a result from \citet{chang:1996:00} to demonstrate that, under mild conditions, the rooted tree topology is identifiable under a non-stationary, strand-symmetric Markov model from the pairwise distributions of states at terminal nodes.

We first need to also import an equivalence relation between tree topologies. Let $\mathcal{T}_1=(S_1,E_1)$ and $\mathcal{T}_2=(S_2,E_2)$ be trees with the same set of terminal nodes $T$. We say that $\mathcal{T}_1$ and $\mathcal{T}_2$ are {\it equivalent} if there is a bijective ``relabelling'' function from $S_1$ to $S_2$ that is the identity function for the terminal nodes and such that the edges $E_2$ are obtained by applying the function to the nodes in the edges $E_1$. That is, the topologies $\mathcal{T}_1$ and $\mathcal{T}_2$ are equivalent if they are the same up to a possible relabeling of internal nodes.

The following is Proposition~3.1 in \citet{chang:1996:00}:

\begin{prop}\label{toporecon}
Consider a family of Markov models satisfying the following conditions:
\begin{enumerate}
\item The edge transition matrices are invertible and not equal to a permutation matrix.
\item There is a node $v$ with $\pi^v(i)>0$ for each $i\in\mathcal{C}$, that is, each character state has positive marginal probability at $v$.
\end{enumerate}
Then the unrooted topology is identifiable from the joint distributions of character states at pairs of terminal nodes. That is, if two models in the family induce the same pairwise distributions of character states at their terminal nodes, then the topologies of those two models must be equivalent.
\end{prop}

As stated in \citet{chang:1996:00}, Proposition~\ref{toporecon} follows by combining a result in \citet{buneman1971recovery} with the {\it additive} function defined on pairs of nodes as
\begin{align*}
f(\{r,s\}) = -\log\det P^{rs} - \log\det P^{sr}.
\end{align*}
The function is additive in the sense that $f(\{r,s\})=\sum_{e\in\mathrm{path}(\{r,s\})}f(e)$, where $\mathrm{path}(\{r,s\})$ is the set of edges joining any two nodes $r$ and $s$. Chang attributed the $\log\det$ distance to \citet{barry:1987:00} and \citet{cavender1987invariants}. It is unclear whether he was aware that the above definition of $f$ is equivalent to the paralinear genetic distance measure introduced by \citet{lake1994reconstructing}. In any case we will prefer to formulate this function almost equivalently as
\begin{align}
f(\{r,s\}) = -\log\det\left(P^{rs} P^{sr}\right),\label{equivform}
\end{align}
because 
\begin{align*}
\det\left(P^{rs} P^{sr}\right) = \det \left(\left(\diag\pi^r\right)^{-1} \left(P^{sr}\right)^\intercal\diag\pi^s P^{sr}\right) \geq 0,
\end{align*}
where $\pi^r$ and $\pi^s$ are the marginal probability vectors at $r$ and $s$ respectively, so that the sign of $\det P^{rs}$ becomes irrelevant. The result in \citet{buneman1971recovery} shows that if the values of an additive function are known between the pairs of all terminal nodes, then under the conditions of Proposition~\ref{toporecon}, the topology and the value of the function on all edges of the topology are determined. As $f$ can be calculated from pairwise distributions between terminal nodes, so the topology and the value of $f$ on every edge can be determined.

\begin{cor}\label{rootedtoporecon}
Suppose the evolutionary tree has nodes of degree three or one, with the exception of one special node which we designate the root. The root can have degree one, two, or three. Assume that
\begin{enumerate}[label=(\Alph*),ref=\Alph*]
\item $\pi^m$ is compositionally asymmetric for every internal node $m$.\label{freig}
\end{enumerate}
Assume also that for each edge $\{u,v\}$, where $v$ is further from the root than $u$,
\begin{enumerate}[label=(\Alph*),ref=\Alph*,resume]
\item $P^{uv}$ is invertible,\label{frinv}
\item $P^{uv}$ is not a permutation matrix, and\label{frperm}
\item $P^{uv}$ is strand-symmetric.\label{frstsy}
\end{enumerate}
Then the rooted topology of the tree is recoverable from pairwise distributions of states at terminal nodes. That is, if two models in the family induce the same pairwise distributions of character states at the terminal nodes, then the topologies recovered by those two models must be equivalent.
\end{cor}

\begin{proof}
Under assumptions \ref{freig}, \ref{frinv}, and \ref{frstsy}, Lemma~\ref{cherry} can be applied to any pair of terminal nodes $a$ and $b$ to obtain $\pi^s$ and $P^{sa}$, where $s$ is the most recent common ancestor of $a$ and $b$, up to a consistent permutation of their elements and rows. We can then calculate $P^{as}=\left(\diag\pi^a\right)^{-1} \left(P^{sa}\right)^\intercal\diag\pi^s$, and so $f(\{s,a\})$. Note that the value of $f(\{s,a\})$ is independent of the chosen permutation of $P^{sa}$ by \eqref{equivform}.

Now fix $a$ and calculate $c=\argmax_{b\in T} f^b(\{s,a\})$ where $f^b(\{s,a\})$ is the value of $f$ for $a$ and the most recent common ancestor of $a$ and $b$, and $T$ are the terminal nodes. The root is then the most recent common ancestor of $a$ and $c$.

By assumptions \ref{freig}, \ref{frinv}, and \ref{frperm}, Proposition~\ref{toporecon} applies and the unrooted topology is identifiable. As for Proposition~\ref{toporecon}, we can also determine the value of $f$ for every edge in the unrooted topology. As the most recent common ancestor of $a$ and $c$ must lie on the path from $a$ to $c$, $f^c(\{s,a\})$ then gives us the address of the root on that path. That is, we have determined the rooted topology and the value of $f$ for every edge in the topology.
\end{proof}

\subsection{Pairwise Distributions Determine the Model}

We now approach the main result in Theorem~\ref{fullrecon}. We show that a non-stationary, strand-symmetric model is fully identifiable given pairwise joint state distributions between all taxa. Up to this point we have made limited assertions about the labelling of states at internal nodes, but we will now make stronger assumptions to remove this ambiguity. One value of knowing this labelling is that it will allow us to fit continuous-time models. 

Start by defining two sets of matrices, the first of which is previously defined in \citet{chang:1996:00}.

\begin{dfn}
A set of matrices $\mathcal{M}$ is {\it reconstructible from rows} if for each $M\in\mathcal{M}$ and each permutation matrix $R\neq I$, we have $RM \notin \mathcal{M}$.
\end{dfn}

As we mentioned in the introduction, the assumption in \citet{chang:1996:00} that every transition probability matrix belong to a class of matrices that is reconstructible from rows is used to identify the labelling of states at internal nodes, which otherwise would not be identifiable.

\begin{example}
To illustrate the abstract concept of a set of matrices that is reconstructible from rows, we borrow an example from \citep{chang:1996:00}. We say that a square matrix $P$ is {\it diagonal largest in column} (DLC) if $P(j,j)>P(i,j)$ for all $i\neq j$. The set of DLC matrices is reconstructible from rows. 
\end{example} 

While DLC serves as a good means for conceptualising a set of matrices that is reconstructible from rows, it has also been shown to be useful as an empirical test for model identifiability \citep{kaehler2015genetic}. As remarked in \citet{chang:1996:00}, it is not the only set of matrices that is reconstructible from rows but might be a reasonable assumption. Unfortunately, as branch lengths increase, the assumption of DLC becomes less reasonable.

\begin{example}\label{sympexample}
Under very weak assumptions, as branch length increases the transition probability matrix on that branch must tend towards its stationary limit. That is, all its rows must tend towards equality. Rows that are closer to equality are in some sense less likely to satisfy DLC. To illustrate the point,
\begin{align*}
P &= \left(\begin{matrix}0.5&0.3&0.2\\0.1&0.4&0.5\\0.1&0.1&0.8\end{matrix}\right)\quad\text{is DLC}.\\
P^2 &= \left(\begin{matrix}0.3&0.29&0.41\\0.14&0.24&0.62\\0.14&0.15&0.71\end{matrix}\right)\quad\text{is not DLC}.
\end{align*}
\end{example}

However, it is an important observation that there is no way of permuting the rows of the second matrix in the above example so that it is DLC. We are therefore motivated to coin the following term.

\begin{dfn}
A set of matrices $\mathcal{S}(\mathcal{M})$ is {\it sympathetic} to a set of matrices $\mathcal{M}$ if it contains all matrices $M$ such that for every permutation matrix $R\neq I$, we have $RM \notin \mathcal{M}$.
\end{dfn}

Note that if $\mathcal{M}$ is reconstructible from rows then $\mathcal{M}\subseteq\mathcal{S}(\mathcal{M})$. To slightly abuse the terminology, sympathetic matrices are useful because they provide matrices that might not be reconstructible from rows, but that will not contradict the ordering of states of internal nodes that is suggested by matrices that are. 

\begin{example}The second matrix in Example~\ref{sympexample} is not DLC but is from the set of matrices that is sympathetic to the set of DLC matrices.
\end{example}

The following theorem extends and restricts Theorem~4.1 in \citet{chang:1996:00}. It makes an additional assumption of strand symmetry, but in so doing is able to accommodate rooted topologies. Motivated by Remark~4 in \citet{chang:1996:00}, it relaxes the assumption that every edge must be reconstructible from rows. Theorem~4.1 in \citet{chang:1996:00} allows {\it degenerate} topologies, where nodes are allowed to have degree greater than three. We assume nondegenerate topologies for ease of exposition and because we will only make that assumption later anyway when we prove that the model can be statistically consistently estimated, as Chang does in his Theorem~5.1. We note that extension of the following theorem to degenerate topologies should be straightforward.

\begin{thm}\label{fullrecon}
Assume the conditions for Corollary~\ref{rootedtoporecon}. Define $\mathcal{M}$ as a set of matrices that is reconstructible from rows and $\mathcal{S}(\mathcal{M})$ as the set of matrices that are sympathetic to $\mathcal{M}$. Further assume that 
\begin{enumerate}[label=(\Alph*),ref=\Alph*,start=5]
\item there exists a path from every internal node $m$ to a terminal node such that for each edge $\{u,v\}$ in the path, where $u$ is closer to $m$ than $v$, $P^{uv}\in\mathcal{M}$, and that\label{frrfr}
\item $P^{uv}\in\mathcal{S}(\mathcal{M})$ for every edge $\{u,v\}$ in the topology\label{frwrfr}.
\end{enumerate}
Then the full model is identifiable. That is, the topology and all of the transition probability matrices are uniquely determined by the joint distribution of character states at the terminal nodes of the tree.
\end{thm}

\begin{remark}
The relaxation of the assumption that every transition probability matrix must be reconstructible from rows to the assumption that every matrix must be sympathetic to those which are might not seem important, but in fact it greatly increases the range of models to which the theorem applies. By allowing sympathetic transition probability matrices, we are in effect allowing long branches in the tree as long as there are sufficient short branches to consistently reconstruct the labelling of states at internal nodes.
\end{remark}

\begin{proof}[Proof of Theorem~\ref{fullrecon}]
We have from Corollary~\ref{rootedtoporecon} the full rooted topology. It is enough to then determine a transition probability matrix for each edge. We will proceed by induction.

If the topology has two terminal nodes and one node is the root then the full model is trivially determined.

If the topology has two terminal nodes and neither node is the root, then by assumptions \ref{freig}, \ref{frinv}, and \ref{frstsy} we can apply Lemma~\ref{cherry} to recover the model up to a permutation of labels of internal states. At least one transition probability matrix must be reconstructible from rows by assumption~\ref{frrfr} and the other at worst will not give an alternative ordering by assumption~\ref{frwrfr}, so the ordering of states at the internal node can be recovered and the full model is determined.

If the topology has more than two terminal nodes, choose an arbitrary terminal node $a$ and denote its neighbouring node $m$. We treat two cases separately.

If $a$ is not the root, choose another terminal node $b$ such that $m$ is the most recent common ancestor of $a$ and $b$. By assumptions \ref{freig}, \ref{frinv}, and \ref{frstsy}, we can apply Lemma~\ref{cherry} to $a$ and $b$ to determine $P^{ma}$ up to a permutation of labels of internal states. By assumption~\ref{frwrfr} we can deduce whether $P^{ma}\in\mathcal{M}$, and if it is we can infer the correct permutation to recover $P^{ma}$. Otherwise, restart this iteration with a different $a$.

If $a$ is the root, choose two other terminal nodes $b$ and $c$ such that the most recent common ancestor of $b$ and $c$ is $m$. Again by assumptions \ref{freig}, \ref{frinv}, and \ref{frstsy}, we can apply Lemma~\ref{cherry} to $b$ and $c$ to obtain a stochastic matrix $V$ and a diagonal matrix $U$ such that $P^{mb}=RV$ and $\diag \pi^m=RUR^\intercal$ for some permutation matrix $R$. Then $P^{ma}=RU^{-1}\left(V^{-1}\right)^\intercal J^{ba}$. Once more if $P^{ma}\in\mathcal{M}$, we infer the correct matrix $R$ and we know $P^{ma}$. Otherwise, restart this iteration with a different $a$.

So we are able to determine $P^{ma}$ for a terminal node $a$ and its neighbour $m$. The induction step is then to remove the edge $\{m,a\}$ from the tree, and make $m$ a terminal node in two subtrees, or one subtree if $m$ happens to be the root. For any subtree, $J^{im}=J^{ia}\left(P^{ma}\right)^{-1}$ for any terminal node of the subtree $i$. This step can be repeated until all subtrees have two taxa and the full model is recovered on all edges.
\end{proof}

\subsection{Continuous-Time Models}

We will now show that Theorem~\ref{fullrecon} is still applicable in almost all cases if we further restrict the model to be continuous-time. We state this explicitly as some care is required because the matrix logarithm can have multiple roots. Continuous-time models provide more information than discrete-time models, for instance for questions concerning genetic distance and relative rates of evolution.

\begin{thm}\label{continuousfullrecon}
Make the assumptions of Theorem~\ref{fullrecon}. Additionally assume that for every $\{u,v\}\in E$, where $v$ is further from the root than $u$, there exists a unique mapping $Q^{uv}=\log P^{uv}$ where every off-diagonal element of $Q^{uv}$ is non-negative. Also replace assumption~\ref{frstsy} with the assumption that $Q^{uv}$ is strand-symmetric. Then the full model is identifiable from the joint distribution of states at the terminal nodes of the tree.
\end{thm}

\begin{remark}
The restriction that the mapping $\log P^{uv}$ be unique may seem like a barrier to implementation of this model. It is difficult to imagine a sufficiently general parametrisation of $Q^{uv}$ that would guarantee this property. However, empirical evidence suggests that it may be better to beg forgiveness than to ask permission \citep{kaehler2015genetic}. Software exists for checking whether a $4\times 4$ Markov generator enjoys this property, and it was found that it was rarely necessary to actually enforce this constraint for the data sets considered in \citet{kaehler2015genetic}.
\end{remark}

\begin{proof}
Again set
\begin{align*}
S=\left(\begin{matrix}
0 & 0 & 0 & 1\\
0 & 0 & 1 & 0\\
0 & 1 & 0 & 0\\
1 & 0 & 0 & 0
\end{matrix}\right).
\end{align*}
Note that if $Q=SQS$ for a square matrix $Q$, then $P=SPS$ for $P=\exp Q$. This follows from the expansion $P=\sum_{i=1}^\infty (n!)^{-1}Q^n$ and that $Q^n=SQ^nS$ because for example $SQ^nS=SQSSQS\ldots SQS$. The assumption that $Q^{uv}$ is strand-symmetric therefore also ensures that assumption~\ref{frstsy} of Theorem~\ref{fullrecon} is satisfied.

By Theorem~\ref{fullrecon}, the mapping from the joint distribution of states at the terminal nodes of the tree to the root marginal probability distribution, the tree topology, and the transition probability matrices $P^{uv}$ is unique, so by assumption this property is also enjoyed by the transition rate matrices $Q^{uv}$.
\end{proof}

\section{Reconstruction from Data: Consistency of Maximum Likelihood}\label{consistency_of_the_full_model}

As an application of the above identifiability results, we shall now prove that the above non-stationary, strand-symmetric models can be statistically consistently estimated from multiple sequence alignments. That is, as the length of the alignment increases, the inferred model must in some sense converge to the correct model.

We will need to use the following definition from \citet{chang:1996:00}, where $\overline{\mathcal{M}}$ denotes the closure of the set of matrices $\mathcal{M}$ under the Euclidean metric.

\begin{dfn}
We say that a set of matrices $\mathcal{M}$ is {\it strongly reconstructible from rows} if, for each $M\in\mathcal{M}$ and each permutation matrix $R\neq I$, we have $RM\notin\overline{\mathcal{M}}$.
\end{dfn}

We also define the following set, where $B_\epsilon(X)$ is the open ball of radius $\epsilon$ centred at $X$ under the Euclidean metric.

\begin{dfn}
A set of matrices $\mathcal{S}_\epsilon(\mathcal{M})$ is {\it strongly sympathetic} to a set of matrices $\mathcal{M}$ if it contains all matrices $M$ such that for every permutation matrix $R\neq I$, we have $B_\epsilon(RM)\cap\mathcal{M}=\varnothing$ for some $\epsilon>0$.
\end{dfn}

\begin{thm}\label{consistency}
Let $\{\mathbb{P}_\theta\}_{\theta\in\Theta}$ be a set of Markov models on trees that have a fixed set of terminal nodes. Suppose the models satisfy the assumptions of Theorem~\ref{fullrecon}, with the exception of assumptions \ref{frrfr} and \ref{frwrfr}, which we restrict slightly to instead assume that from every internal node $m$ there exists a path to a terminal node such that for every edge $\{u,v\}$ in the path, $P^{uv}\in\mathcal{M}$, where $v$ is further from $m$ than $u$, and for every edge $\{u,v\}\in E$, $P^{uv}\in\mathcal{M}\cup\mathcal{S}_\epsilon\left(\mathcal{M}\right)$, $\mathcal{M}$ is strongly reconstructible from rows, and $\mathcal{S}_\epsilon\left(\mathcal{M}\right)$ is strongly sympathetic to $\mathcal{M}$.

Then the method of maximum likelihood consistently recovers the topology, root marginal probability distribution, and edge transition probability matrices. That is, for any $\theta\in\Theta$, let $\hat{\theta}_n$ denote the maximum likelihood estimate based on $n$ independent observations $\{X_T^i\}_{i\in\{1,\ldots,n\}}$ of character states at the terminal nodes of the tree, then
\begin{align*}
\hat{\theta}_n\cvgas\theta\quad\text{as}\quad n\rightarrow\infty.
\end{align*}
\end{thm}

\begin{thm}\label{continuousconsistency}
Let $\{\mathbb{P}_\theta\}_{\theta\in\Theta}$ be a set of Markov models on trees that have a fixed set of terminal nodes. Suppose the models satisfy the assumptions of Theorem~\ref{consistency}. Additionally assume that for every $\{u,v\}\in E$, where $v$ is further from the root than $u$, that there exists a unique mapping $Q^{uv}=\log P^{uv}$, where every off-diagonal element of $Q^{uv}$ is non-negative. Also replace assumption~\ref{frstsy} from Theorem~\ref{fullrecon} with the assumption that $Q^{uv}=SQ^{uv}S$.

Then the method of maximum likelihood consistently recovers the topology, root marginal probability distribution, and edge transition rate matrices. That is, for any $\theta\in\Theta$, let $\hat{\theta}_n$ denote the maximum likelihood estimate based on $n$ independent observations $\{X_T^i\}_{i\in\{1,\ldots,n\}}$ of character states at the terminal nodes of the tree, then $\hat{\theta}_n\cvgas\theta$ as $n\rightarrow\infty$.
\end{thm}

The following is a restatement of Lemma~5.1 in \citet{chang:1996:00}, where a sketch of its proof is provided.

\begin{lem}\label{changwald}
Let $\mathcal{X}$ be a finite set and let $\{\mathcal{P}_\theta\}_{\theta\in\Theta}$ be a set of probability distributions on $\mathcal{X}$, where the closure $\overline{\Theta}$ of $\Theta$ is a compact subset of a metric space. Let $\{X_i\}_{i\in\integers}$ be independent and identically distributed random variables (or vectors) with probability distribution $\mathcal{P}_{\theta_0}$ for some $\theta_0\in\Theta$. Assume the identifiability condition
\begin{align*}
\mathcal{P}_\theta\neq\mathcal{P}_{\theta_0}\quad\text{for each}\quad\theta\in\overline{\Theta}\quad\text{with}\quad\theta\neq\theta_0.
\end{align*}
Suppose that for each $x\in\mathcal{X}$ the function $\theta\mapsto\mathcal{P}_\theta(x)$ is continuous on $\overline{\Theta}$ and let $\hat{\theta}_n=\hat{\theta}_n(X_1,\ldots,X_n)$ maximise the log likelihood $\sum_{i=1}^n\log\mathcal{P}_\theta(X_i)$ over $\theta\in\overline{\Theta}$. Then under $\mathcal{P}_{\theta_0}$, $\hat{\theta}_n\cvgas\theta_0$ as $n\rightarrow\infty$.
\end{lem}

\begin{proof}[Proof of Theorem~\ref{consistency}]
The proof of Theorem~5.1 in \citet{chang:1996:00} carries to this context with only slight modification. We need to show that the conditions of Lemma~\ref{changwald} are implied by the assumptions of Theorem~\ref{consistency} with $\mathcal{X}=\{A,C,G,T\}^{|T|}$ and $\mathcal{P}_\theta(A)=\mathbb{P}_\theta(X_T\in A)$ for $A\subseteq\mathcal{X}$.

For $\theta,\theta_0\in\Theta$, Theorem~\ref{fullrecon} gives that $\mathcal{P}_\theta\neq\mathcal{P}_{\theta_0}$ if $\theta\neq\theta_0$. It remains to show for $\theta_0\in\Theta$ and $\theta\in\overline{\Theta}$ that $\mathcal{P}_\theta\neq\mathcal{P}_{\theta_0}$ if $\theta\neq\theta_0$.

The proof of Theorem~5.1 in \citet{chang:1996:00} shows that the assumptions that $P^{uv}$ is invertible and not a permutation matrix for every edge $\{u,v\}$ must be satisfied for a model $\theta\in\overline{\Theta}$ if $\mathcal{P}_\theta=\mathcal{P}_{\theta_0}$ for some $\theta_0\in\Theta$, by virtue of the form of $\mathcal{P}_{\theta_0}$. The set of strand-symmetric matrices is closed, so the assumption that each such $P^{uv}$ is strand-symmetric must be satisfied for all $\theta\in\overline{\Theta}$. Finally, if compositional asymmetry is violated at any internal node for a $\theta\in\overline{\Theta}$, then for at least one pair of terminal nodes $a$ and $b$, the eigendecomposition of $S\left(J^{ab}\right)^{-1}SJ^{ab}$ will have at least one repeated eigenvalue. That is, if $\mathcal{P}_\theta=\mathcal{P}_{\theta_0}$ for some $\theta_0\in\Theta$, then for $\theta$ we must have that all internal nodes satisfy compositional asymmetry.

So far we have shown that if $\mathcal{P}_\theta=\mathcal{P}_{\theta_0}$ for $\theta_0\in\Theta$ and $\theta\in\overline{\Theta}$ then the assumptions of Corollary~\ref{rootedtoporecon} are satisfied by both models and so they must have the same rooted topology.

It remains to show that if $\mathcal{P}_\theta=\mathcal{P}_{\theta_0}$ for $\theta_0\in\Theta$ and $\theta\in\overline{\Theta}$ then the transition probability matrices for $\theta_0$ are the same as for $\theta$. We will show that the inductive steps of Theorem~\ref{fullrecon} still apply in this context. As the assumptions of Corollary~\ref{rootedtoporecon} are satisfied for $\theta$, Lemma~\ref{cherry} can be safely applied, so the only remaining question is the ordering of rows of the transition probability matrices.

Take a transition probability matrix $P$ from the model $\theta\in\overline{\Theta}$ and the matrix $P_0$ on the corresponding edge from $\theta_0\in\Theta$. Firstly assume $P_0\in\mathcal{M}$. For some permutation matrix $R$, $P=RP_0$, but if $R\neq I$, $RP_0\notin\overline{\mathcal{M}}$ as $\mathcal{M}$ is strongly reconstructible from rows. We must also check that $RP_0\notin\overline{\mathcal{S}_\epsilon(\mathcal{M})}$. By construction, $\overline{\mathcal{S}_\epsilon(\mathcal{M})}\subseteq\mathcal{S}(\mathcal{M})$ and $RP_0\notin\mathcal{S}(\mathcal{M})$, so $RP_0\notin\overline{\mathcal{S}_\epsilon(\mathcal{M})}$ and $P=P_0$. Next assume that $P_0\in\mathcal{S}_\epsilon(\mathcal{M})$. Again $P=RP_0$, but $RP_0\notin\overline{\mathcal{M}}$. In this case the order of rows of $P$ is decided by the order of those on neighbouring edges, and must be the same for $\theta$ and $\theta_0$.

It is therefore possible to modify the proof of Theorem~\ref{fullrecon} such that it ensures that for $\theta_0\in\Theta$ and $\theta\in\overline{\Theta}$, $\mathcal{P}_\theta\neq\mathcal{P}_{\theta_0}$ if $\theta\neq\theta_0$

We are therefore able to apply Lemma~\ref{changwald} and the model can be estimated consistently under maximum likelihood.
\end{proof}

\begin{proof}[Proof of Theorem~\ref{continuousconsistency}]
This theorem follows from the proofs of Theorems \ref{continuousfullrecon} and \ref{consistency}. The only difficulty is to check whether there might exist a $Q$ which is part of the model $\theta$ and a corresponding $Q_0$ which is a part of the model $\theta_0$ such that $\mathcal{P}_\theta=\mathcal{P}_{\theta_0}$ but $Q\neq Q_0$. This could only occur if the transition probability matrix $P$ from the same edge and model $\theta$ mapped to multiple valid Markov generators. As for the proof of Theorem~\ref{consistency}, all transition probability matrices are uniquely determined by $\mathcal{P}_\theta$ and as $\mathcal{P}_\theta=\mathcal{P}_{\theta_0}$ we must have that $P=P_0$ for the corresponding $P_0$ in $\theta_0$, so by assumption $Q$ must be unique and $Q=Q_0$.
\end{proof}

\section{Discussion}\label{conculding_remarks}

Assumption~\ref{freig} of Corollary~\ref{toporecon} implies that the process under consideration is non-stationary and we subsequently find that we can recover the root. We further show in example~\ref{statex} that in the stationary, non-reversible case it is not possible to recover the root. This seems to fit with the collective empirical evidence of \citet{yang1995use}, \citet{huelsenbeck2002inferring}, and \citet{yap2005rooting}. That is, in some cases a non-stationary process can recover the root, whereas a stationary, non-reversible process cannot. Assumption~\ref{freig} is also the source of a subtle and interesting contradiction. It can be shown under mild assumptions that the product of arbitrary stochastic matrices will eventually have almost identical rows as the number of terms in the product increases \citep[][Section~4.3]{seneta2006non}. It is not difficult to show that the rows must be the stationary distribution of the product, and that if the terms are strand symmetric then the product must be strand symmetric. So for any reasonably long history of strand-symmetric processes, we must have that the marginal probabilities are also strand-symmetric. The conclusion is that for our model to work, the process should probably have been strand asymmetric prior to the root of the tree, and strand symmetric thereafter. An alternative interpretation is that our model is just a model, and that we could start by assuming that the root probabilities are sufficiently asymmetric for the process to be non-stationary, and then that the process is sufficiently strand-symmetric for the model to be approximately correct. In either case assumption~\ref{freig} of Corollary~\ref{toporecon} could be an interesting tool for probing the limits of inference of this model.

In the spirit of \citet{chang:1996:00} we have attempted to provide a constructive proof, and hope that the results here presented will enable the implementation of new phylogenetic methods. There are several options for such an implementation. \citet{mossel2006learning} provide an algorithm for directly applying the concepts in \citet{chang:1996:00} to learn the full model parameters using spectral methods. It would be relatively straightforward to adapt that approach to the proof to Theorem~\ref{fullrecon}. Alternatively, progress is being made towards fitting models that are {\it heterogeneous across lineages} \citep{jayaswal2011reducing,jayaswal2014mixture}. These methods could be adapted to our setting. Also, the method of \citet{yap2005rooting} where a homogeneous process is fitted to the whole tree would be trivial to implement for a strand symmetric process, but could be done with a solid theoretical basis and new insight into the limits of inference of such a model.

In addition to these immediate applications, this work provides a foundation for theoretical developments where a non-stationary model can be fully recovered for a two-taxon rooted topology. The ideas presented here could be extended to any such model, not just the strand-symmetric one on which we focus. Another possible direction for future theoretical development is a model that is rate-heterogeneous amongst alignment columns. We have ignored this possibility, both because it is possible to work around it with careful site classification, as in \citet{yap2005rooting}, and because recent results in non-stationary phylogenetic processes have shown that the general time-reversible model is less biased than the usual rate-heterogeneous general time-reversible model in comparison to a general non-stationary process in some circumstances \citep{kaehler2015genetic}. Another possible line of enquiry is to refine the sufficient conditions of Theorem~\ref{fullrecon}. We know that a stationary process is not identifiable for a rooted topology, but that our non-stationary strand-symmetric process is. Our class of models is slightly narrower than the class of non-stationary, strand-symmetric processes, however, so it may be possible to broaden our sufficient assumptions.

\subsection*{Acknowledgements}
I am grateful to Teresa Neeman, Gavin Huttley, Von Bing Yap, Michael Roper, and John Trueman for their helpful comments. I am particularly indebted to Teresa for suggesting Corollary~\ref{rootedtoporecon}, which greatly simplified the proof of Theorem~\ref{fullrecon}.

This work was supported by the National Health and Medical Research Council [grant number APP1085372].

\bibliographystyle{chicagoa}
\bibliography{refs}
\end{document}